\documentclass[prl,twocolumn,superscriptaddress,aps,10pt]{revtex4-1}
\usepackage[utf8]{inputenc}
\usepackage[english]{babel}
\usepackage{amsmath}
\usepackage{amsfonts}
\usepackage{amssymb}
\usepackage{graphicx}
\usepackage{times}
\usepackage{color}
\usepackage[usenames,dvipsnames]{xcolor}

\usepackage{float}
\usepackage{bbm}
\usepackage{amsthm}
\definecolor{myurlcolor}{rgb}{0,0,0.7}
\definecolor{myrefcolor}{rgb}{0,0,0.7}
\usepackage[unicode=true,pdfusetitle,bookmarks=true,bookmarksnumbered=false,bookmarksopen=false,breaklinks=false,pdfborder={0 0 0},backref=false,colorlinks=true,linkcolor=myrefcolor,citecolor=myurlcolor,urlcolor=myurlcolor]{hyperref}

\DeclareMathOperator{\corr}{corr}
\DeclareMathOperator{\tr}{tr}

\DeclareMathOperator{\diff}{d\!}

\usepackage{color}

\usepackage{mathtools}

\newcommand{\e}{\mathrm{e}}
\newcommand{\ii}{\mathrm{i}}

\newcommand{\expect}[1]{\langle #1 \rangle}

\newtheorem{theorem}{Theorem}
\newtheorem{lemma}[theorem]{Lemma}

\begin{document}

\title{Correlation decay in fermionic lattice systems with power-law interactions at non-zero temperature}
\author{Senaida Hern\'andez-Santana}
\affiliation{ICFO-Institut de Ciencies Fotoniques, The Barcelona Institute of Science and Technology, 08860 Castelldefels (Barcelona), Spain}
\author{Christian Gogolin}
\affiliation{ICFO-Institut de Ciencies Fotoniques, The Barcelona Institute of Science and Technology, 08860 Castelldefels (Barcelona), Spain}
\affiliation{Max-Planck-Institut für Quantenoptik, Hans-Kopfermann-Straße 1, 85748 Garching, Germany}
\author{J. Ignacio Cirac}
\affiliation{Max-Planck-Institut für Quantenoptik, Hans-Kopfermann-Straße 1, 85748 Garching, Germany}
\author{Antonio Ac\'in}
\affiliation{ICFO-Institut de Ciencies Fotoniques, The Barcelona Institute of Science and Technology, 08860 Castelldefels (Barcelona), Spain}
\affiliation{ICREA-Instituci\'o Catalana de Recerca i Estudis Avan\c cats, 08010 Barcelona, Spain}

\begin{abstract}
  We study correlations in fermionic lattice systems with long-range interactions in thermal equilibrium.
  We prove a bound on the correlation decay between anti-commuting operators and generalize a long-range Lieb-Robinson type bound.
  Our results show that in these systems of spatial dimension $D$ with, not necessarily translation invariant, two-site interactions decaying algebraically with the distance with an exponent $\alpha \geq 2\,D$, correlations between such operators decay at least algebraically with an exponent arbitrarily close to $\alpha$ at any non-zero temperature.
  Our bound is asymptotically tight, which we demonstrate by a high temperature expansion and by numerically analyzing density-density correlations in the 1D quadratic (free, exactly solvable) Kitaev chain with long-range pairing.
\end{abstract}

\maketitle

Systems with long-range interactions decaying algebraically (power-law like) with the distance have many fascinating properties setting them apart from systems with merely finite-range or exponentially decaying (short range) interactions.
Very recently, a surge of interest in the properties of these models has lead to a wealth of new insights.
For example, in such systems very quick equilibration \cite{Kastner2011b,Bachelard2013,Kastner2016} and fast spreading of correlations \cite{Richerme2014,Maghrebi2015a}, as well as violations of the area law \cite{Koffel2012a} and very fast state transfer \cite{Eldredge2016} are possible.
Most importantly, they show topological effects and support Majorana edge modes \cite{Vodola2014c,Patrick2016}.
This development is to a large extent a consequence of the fact that such systems can be realized \cite{Porras2004,Deng2005,Hauke2010,Gullans2015} in extremely well controlled experiments with polar molecules \cite{Micheli2006}, ultra-cold ions \cite{Islam2011,Schneide2012,Britton2012,Jurcevic2014,Richerme2014}, and Rydberg atoms \cite{Labuhn2016}.
At the same time, many of the fundamental interactions in nature are actually algebraically decaying, such as dipole-dipole interactions, the van der Waals force, and, last but not least, the Coulomb interaction.

In some cases, realistic systems can be approximately captured by finite-range models, for example in the limit of a tight binding approximation.
The physics of such systems has been at the center of attention of theoretical condensed matter physics.
In particular, it has been proven for finite-range fermionic systems that the correlations between anti-commutating operators decay exponentially at any non-zero temperature \cite{Hastings2004a} and the same holds at zero temperature whenever there is a non-vanishing gap above the ground state \cite{Hastings2006}.
Similarly, arbitrary observables above a threshold temperature in finite-range spin and fermionic systems \cite{Kliesch2014} show exponential decay of correlations.
A similar level of understanding of the correlation decay of truly long-range interacting systems is lacking so far \cite{Patrick2016}, but is no less desirable due to their intriguing properties \cite{Cannas1996,Dyson1969c,Fischer1972,Dutta2001,Porras2004,Deng2005,Hauke2010,Dalmonte2010,Kastner2011b,Koffel2012a,Peter2012,Bachelard2013,Kastner2016,Eldredge2016,Santos2015,Patrick2016}.

The goal of our work is to advance the understanding of the decay of correlations in long-range interacting systems at finite temperature.
Our main result predicts that correlations at non-zero temperature in general two-site interacting fermionic long-range systems of arbitrary spatial dimension decay at least with essentially the same exponent as the interaction strength.
The bound holds in both clean, translation invariant systems and in such with disorder.
This result is based on recent advances \cite{Foss-Feig2014} on the dynamical spreading of correlations in long-range interacting systems.
We demonstrate that our bound is asymptotically tight by means of a high temperature expansion and by numerical simulations of a 1D Kitaev chain of fermions with long-range p-wave pairing at finite temperature, whose ground state phase diagram has been extensively studied \cite{Vodola2016b,Vodola2014c}.
As our bound (which holds for all non-zero temperatures) can be asymptotically saturated already at arbitrarily high temperature and as correlations typically do not decay faster at low temperatures, our result suggests the absence of phase transitions in such models that impact the asymptotic decay behavior of correlations.

\paragraph{Setting and notation.}

We study the correlations and their decay behavior in quantum many-body systems in thermal equilibrium at finite temperature $T$.
We focus on systems of spinless fermions in which for each site $i \in \{1\dots,L\}$ we have a fermionic creation $a_i^\dagger$ and an annihilation operator $a_i$ that satisfy the anti-commutation relations $\{a_i,a_j^\dagger\} \coloneqq a_i\,a_j^\dagger + a_j^\dagger\,a_i =\delta_{i,j}$ (a generalization to spin-full fermions is straight forward).
We denote by $n_i \coloneqq a_i^\dagger a_i$ the particle number operator of site $i$.
For $A$ and $B$ operators on the Fock space we define their correlation coefficient as
\begin{equation}
  \corr(A,B)_\beta \coloneqq \expect{A\,B}_\beta - \expect{A}_\beta\,\expect{B}_\beta ,
\end{equation}
where $\expect{\cdot}_\beta$ is the expectation value in the thermal state
\begin{equation}
  \rho_\beta \coloneqq \e^{-\beta\,H} / \tr(\e^{-\beta\,H})
\end{equation}
at inverse temperature $\beta \coloneqq 1/k_B T$.
We call an operator $A$ even (odd) if it can be written as an even (odd) polynomial of creation and annihilation operators, i.e., if it is a sum of monomials that are all products of an even (odd) number of $a_i$ and $a_i^\dagger$.
Odd operators anti-commute when they have disjoint supports.
Due to the particle number parity super-selection rule Hamiltonians of physical systems are even operators and hence $\expect{A}_\beta = 0$ whenever $A$ is an odd operator.

In what follows, we will mostly be interested in the correlations between operators $A$ and $B$ that are either particle number operators on different sites or odd operators on disjoint regions and how $\corr(A,B)$ decays with the distance of their supports.

Our result is obtained for fermionic system on a hypercubic lattice of dimension $D$ whose Hamiltonian can, for some constant $J$, be written in the form
\begin{equation}
  \label{systems}
  H = \sum_{\kappa,i,j} J_{i,j}^{(\kappa)}\,V_i^{(\kappa)}\,V_j^{(\kappa)} ,
\end{equation}
in terms of normalized operators $V_i^{(\kappa)}$, each acting on their respective site $i$, and coupling coefficients $J_{i,j}^{(\kappa)}$ satisfying $\sum_\kappa J_{i,j}^{(\kappa)} \leq J\,{d_{i,j}}^{-\alpha}$ with $d_{i,j}$ the $L_1$-distance between the sites $i$ and $j$.
Thus, our result holds for fermionic systems with quadratic Hamiltonians as well as non-quadratic ones with two-site interactions.

\paragraph{A general bound on correlation decay in fermionic long-range systems.}

We now derive the main result of this work, a general bound on the algebraic decay of correlations in fermionic systems with long-range interactions at non-zero temperature.
Concretely, for $A,B$ odd operators we obtain a bound on $\corr_\beta(A,B) = \expect{A\,B}_\beta$.
In the special case of quadratic Hamiltonians, like the Kitaev chain we consider later, our bound also yields, via Wick's theorem, a bound on density-density correlations.
Our result is based on two main ingredients:
An integral representation of $\expect{A\,B}_\beta$ that was previously used in \cite{Hastings2004a} and an extension to the fermionic case of a very recently derived Lieb-Robinson-type bound for systems with long-range interactions~\cite{Foss-Feig2014}.

The first ingredient for our proof is the following integral representation \footnote{See Supplemental Material \ref*{proof_of_lemma_integralrepresentation} for details of the proof of Lemma~\ref{lemma:integralrepresentation}.} of the expectation value $\expect{A\,B}_\beta$:
\begin{lemma}[integral representation \cite{Hastings2004a}]
  \label{lemma:integralrepresentation}
  Given a fermionic system at inverse temperature $\beta > 0$ and an
  even Hamiltonian and any two odd operators $A,B$ it holds that
  \begin{equation}\label{eq:integralrepresentation}
      \expect{A\,B}_\beta = \frac{\expect{\{A,B\}}_\beta}{2}
      + \int_0^\infty \frac{\ii}{\beta} \frac{\expect{\{A(t)-A(-t),B\}}_\beta}{\e^{\pi\,t/\beta} - \e^{-\pi\,t/\beta}} \diff t.
  \end{equation}
\end{lemma}

Lieb and Robinson~\cite{Lieb1972} first proved that the propagation of information in quantum spin systems with short-range interactions is characterized by a group velocity bounded by a finite constant, which leads to a light-cone-like causality region.
This results has since been generalized and improved in various aspects \cite{Hastings2004,Nachtergaele11} (see also \cite{Kliesch2013} for a review).
Hastings and Koma~\cite{Hastings2006} proved an upper-bound on the group velocity that grows exponentially in time in systems 
with 
power-law decaying interactions with exponent $\alpha>D$. 
Improving upon this, Gong et al.~\cite{Gong2014} derived a bound for $\alpha > D$, that consists of a exponentially and a power-law like decaying contribution.
Foss-Feig et al.~\cite{Foss-Feig2014} proved a Lieb-Robinson type bound with a group-velocity bounded by a power-law for two-site long-range interacting spin systems with the same form as in Eq.~\eqref{systems} for $\alpha > 2\,D$.
Further, Matsuta et al.~\cite{Matsuta2016} proved a closely related bound for long-range interacting spin systems for all $\alpha > D$.
For $\alpha < D$ energy is no longer extensive and Lieb-Robinson-like bounds can only be achieved \cite{Eisert2013b} when time is rescaled with the system size \cite{Storch2015a}.

For the purpose of our proof, we extend the Lieb-Robinson bound obtained by Foss-Feig et al.~\cite{Foss-Feig2014} to fermionic systems.
Here it takes the form of a bound on the operator norm $\|\ \|$ of the anti-commutator of odd operators:
\begin{lemma}[Lieb-Robinson-like bound for fermionic long-range systems]
  \label{lemma:long-range-lieb-robinson}
  Consider a fermionic system on a hypercubic lattice of dimension $D$.
  Let $\alpha>2\,D$ and $\gamma\coloneqq(1+D)/(\alpha-2\,D)$.
  Assume that the Hamiltonian can be written in the form~\eqref{systems} with $J$ a constant.
  Then, for any two odd operators $A$ and $B$ separated by a distance $l$ there exist constants $c_0$ and $c_1$, independent of the system size, $l$, and $t$, such that
  \begin{equation}\label{eqf05}
    \|\{A(t),B\}\| \leq c_0\,\e^{v\,|t|-l/|t|^{\gamma}} + c_1\,\frac{|t|^{\alpha\,(1+\gamma)}}{l^{\alpha}} ,
  \end{equation}
  with $v \leq 8\,J\,\exp(1)\,2^D$.
\end{lemma}
\noindent The proof of Lemma~\ref{lemma:long-range-lieb-robinson} follows the general strategy of \cite{Foss-Feig2014}.
We explain all necessary technical modifications in \footnote{See Supplemental Material \ref*{ms:proof_of_RL_bound} for details of the proof of Lemma~\ref{lemma:long-range-lieb-robinson}.}.

The main result of this work is that correlations in fermionic systems with two-site long range interactions at non-zero temperature decay at least algebraically with an exponent essentially given by the exponent $\alpha$ of the decay of the long-range interactions:
\begin{theorem}[Power-law decay of correlations]
  \label{theorem:power-law-deca-of-correlations}
  Consider a fermionic system on a $D$ dimensional hypercubic lattice with a Hamiltonian of the form given in~\eqref{systems} with $J$ a constant and $\alpha > 2\,D$.
  For any two odd operators $A,B$, denoting by $l$ the distance between their supports, then for any $0 < \epsilon < 1$
  \begin{equation}
    \| \corr(A,B)_\beta \| \in \mathcal{O}(l^{-(1-\epsilon)\,\alpha}) \quad (l \to \infty) .
  \end{equation}
\end{theorem}
Before we present the proof (which actually yields a concrete bound with calculable prefactors) of this theorem, let us interpret the result.
It says that the correlations between any two odd (and therefore anti-commuting) operators in long-range interacting fermionic systems in thermal equilibrium at non-zero temperature decay at least power-law like at long distances, with an exponent that is arbitrarily close to the exponent $\alpha$ of the long-range interactions.
This holds for systems with an arbitrary spatial dimension $D$ as long as $\alpha \geq 2\,D$.

\begin{proof}[Proof of Theorem~\ref{theorem:power-law-deca-of-correlations}]
  We start by using Lemma~\ref{lemma:integralrepresentation}.
  As $\{A,B\} = 0$ only the second term from Eq.~\eqref{eq:integralrepresentation} is non-zero.
  We split up the integral in this term $I = I_{\leq\tau(l)} + I_{>\tau(l)}$ into an integral $I_{\leq\tau(l)}$ from time zero up to some value $\tau(l)$ (whose dependence on $l$ we will chose later) and the rest $I_{>\tau(l)}$.
  We bound these two integrals separately.
  Using that $|\expect{\{ A(t)- A(-t),B \}} _\beta| \leq 4\,\|A\|\,\|B\|$, we find
  \begin{equation}\label{eqf08}
    |I_{>\tau(l)}| \leq \left| \int_{\tau(l)}^{\infty} \frac{1}{\beta}\,\frac{4\,\|A\|\,\|B\|}{ \e^{\pi\,t/\beta}-\e^{-\pi\,t/\beta}}\,\diff t \right| .
  \end{equation}
  The integral satisfies
  \begin{equation}
    \int_{\tau(l)}^{\infty}\frac{1}{\beta}\,\frac{\diff t}{\e^{\pi\,t/\beta}-\e^{-\pi\,t/\beta}} \leq \frac{\pi^{-1}}{ \e^{\pi\,\tau(l)/\beta} - \e^{-\pi\,\tau(l)/\beta}}
  \end{equation}
  and therefore we have
  \begin{equation}\label{eq1B3}
    |I_{>\tau(l)}| \leq \frac{c_2/\pi}{\e^{\pi\,\tau(l)/\beta}-\e^{-\pi\,\tau(l)/\beta}} 
  \end{equation}
  with $c_2 \coloneqq 4\,\|A\|\,\|B\|$.

  For the second term $I_{<\tau(l)}$ we use that $| \expect{\{A(t),B\}}_\beta| \leq \| \{ A(t),B \} \|$ so that
  \begin{equation}\label{eqf090}
    |I_{<\tau(l)}| \leq \int_0^{\tau(l)} \frac{1}{\beta}\,\frac{\| \{A(t),B\} \| + \| \{ A(-t),B \} \|}{\e^{\pi\,t/\beta} - \e^{\pi\,t/\beta}}\,\diff t .
  \end{equation}
  Next, we apply the Lieb-Robinson-like bound from Lemma~\ref{lemma:long-range-lieb-robinson}, 
  \begin{align}\label{eq1B4}
    |I_{<\tau(l)}|    & \leq \nonumber \frac{2\,c_0}{\beta}\,\int_0^{\tau(l)} \frac{\e^{v\,t-l/t^{\gamma}}}{\e^{\pi\,t/\beta}-\e^{\pi\,t/\beta}}\,\diff t \\
                    & + \frac{2\,c_1}{\beta}\,\frac{1}{l^{\alpha}}\int_0^{\tau(l)} \frac{t^{\alpha\,(1+\gamma)}}{\e^{\pi\,t/\beta}-\e^{\pi\,t/\beta}}\,\diff t .
  \end{align}
  As $(\e^{\pi\,t/\beta}-\e^{\pi\,t/\beta})^{-1}\leq \frac{\beta}{2\,\pi\,t}$ we further have
  \begin{equation}\label{eq1B5}
      |I_{<\tau(l)}| \leq \frac{c_0}{\pi}\,\int_0^{\tau(l)} \frac{\e^{v\,t-l/t^{\gamma}}}{t}\,\diff t + \frac{c_1\,\tau(l)^{\alpha\,(1+\gamma)}}{\pi\,\alpha\,(1+\gamma)\,l^{\alpha}}.
  \end{equation}
  Now, let $g(t) \coloneqq \e^{v\,t}\,\e^{-l/t^{\gamma}}/t$.
  Notice that $g(t)$ is a product of the monotonically increasing function $\e^{vt}$ and the function $h(t) \coloneqq \e^{-l/t^{\gamma}}/t$ which satisfies: (i) it has a local maximum at $t_h^*(l) \coloneqq (\gamma\,l)^{1/\gamma}$; (ii) it is monotonically increasing in $[0,t_h^*(l)]$.
  Therefore, $g(t)$ is also monotonically increasing in $[0,t_h^*(l)]$ so that, provided that $\tau(l) < t_h^*(l)$, we can bound 
  \begin{equation}\label{eq1B7}
    \int_0^{\tau(l)} g(t)\,\diff t \leq g(\tau(l))\,\tau(l) = \e^{v\,\tau(l)-l/\tau(l)^{\gamma}}.
  \end{equation}
  For all $\tau(l) < t_h^*(l)$ we hence have the upper-bound
  \begin{equation}\label{eq1B8}
    |I_{<\tau(l)}| \leq \frac{c_0}{\pi}\,\e^{v\,\tau(l)-l/\tau(l)^{\gamma}} + \frac{c_1}{\pi}\,\frac{1}{\alpha\,(1+\gamma)}\,\frac{\tau(l)^{\alpha\,(1+\gamma)}}{l^{\alpha}} .
  \end{equation}

  It remains to find a good choice for $\tau(l)$.
  The function $\tau$ must grow unbounded with increasing $l$ in order for the right hand side of Eq.~\eqref{eq1B3} to go to zero and, at the same time, it must not grow too fast, so that $\tau(l) < t_h^*(l)$ is satisfied and the right hand side of Eq.~\eqref{eq1B8} goes to zero for large $l$.
  We take $\tau(l) = (l/v)^{\frac{1}{\gamma +1}}\,l^{-\eta}$ with $\eta \in ]0,1/(\gamma+1)[$.
  This yields that for all such $\eta$
  \begin{equation}
    \begin{split}
      & \| \expect{ A\,B }_\beta \| \leq \frac{c_0}{\pi}\,\e^{v^{\frac{\gamma}{\gamma +1}}\,l^{\frac{1}{\gamma +1}}\,(l^{-\eta}-l^{\gamma\,\eta})}\\
      &+ \frac{c_1/\pi}{\alpha\,(1+\gamma)\,v^{\alpha}}\,l^{-\eta\,(1+\gamma)\,\alpha} 
+ \frac{c_2/\pi}{\e^{\pi\,v^{\frac{-1}{\gamma +1}}\,l^{\frac{1}{\gamma +1}-\eta}/\beta}-1} .
    \end{split}
  \end{equation}
  As $\gamma$ and $\eta$ are positive and $\eta < 1/(\gamma+1)$, both the first and the last term decay super-algebraically for large $l$.
  The dominating term is thus the middle term, which implies the result as stated, where $\epsilon = 1-\eta\,(1+\gamma)$.
\end{proof}
We remark that we were not able to prove Theorem~\ref{theorem:power-law-deca-of-correlations} from the other Lieb-Robinson bounds for systems with long-range interactions.
In particular, when using the bound from \cite{Matsuta2016} that is valid for all $\alpha>D$, the term corresponding to the first term in Eq.~\eqref{eq1B5} diverges because of the behavior of the integrand in the limit $t \to 0$.
It remains open whether the restriction to $\alpha > 2\,D$ in our result is an artifact of our proof technique or whether there is a physical reason, at least the point $\alpha = 2\,D$ was identified to be a special case in \cite{Bachelard2013}.

\paragraph{Kitaev chain with long-range interactions.}

In the numerical part of this work we consider a generalization of the fermionic Kitaev chain \cite{Kitaev2001} with long-range p-wave pairing of size $L$, whose Hamiltonian $H \coloneqq H_{\mathrm{FR}} + H_{\mathrm{LR}}$ consists of a finite-range (nearest neighbor) part
\begin{equation} \label{eq:hamiltonianSR} H_{\mathrm{FR}} \coloneqq -t
  \sum_{i=1}^L \big(a_i^\dagger\,a_{i+1} + \text{h.c.} \big)
  - \mu \sum_{i=1}^L \big(n_i-1/2 \big)
\end{equation}
with tunneling rate $t$ and chemical potential $\mu$, and a power-law decaying long-range pair-creation/pair-annihilation term
\begin{equation} \label{eq:hamiltonianLR} H_{\mathrm{LR}} \coloneqq
  \frac{\Delta}{2} \sum_{i=1}^L \sum_{j=1}^{L-1}d_j^{-\alpha}\big(a_i\,a_{i+j} + a_{i+j}^\dagger\,a_i^\dagger \big) ,
\end{equation}
where $d_j \coloneqq \min(j,L-j)$, $\Delta$ is the coupling strength, and $\alpha$ the coupling exponent \cite{Vodola2014c}.
Whenever $L$ is finite, we consider a closed chain with anti-periodic boundary conditions, i.e., for $i>L$ we set $a_{i} \coloneqq -a_{i \mod L}$ as otherwise the long-range term vanishes due to the fermionic commutation relations \footnote{See Supplemental Material \ref*{appSRPBC} for a justification of that choice.}.
As in \cite{Vodola2014c}, in the remainder of this work, we consider the case $\Delta = 2\,t = 1$.
This model has a rich ground state phase diagram with two critical points at $\mu=\pm1$ \cite{Vodola2016b,Vodola2014c}.

The model described above falls into the class of so-called, quadratic, free, or non-interacting models.
Their Hamiltonians can be written as
$
H = \sum_{i,j} c_i^\dagger\,h_{ij}\,c_j
$
where $\vec{c}\coloneqq(a_1,a_1^\dagger, \dots, a_m,a_m^\dagger)$ and the Hamiltonian matrix $h$ is hermitian.
By diagonalizing $h = U^\dagger\,D\,U$ it can then be brought into the form
$
H = \sum_{i} b_i^\dagger\,D_{ii}\,b_i ,
$
with $\vec{b} \coloneqq U\,\vec{c}$.
From this normal-mode decomposition one can compute the elements $\corr(b_j, b_k^\dagger)_\beta$ of the covariance matrix of the thermal state and, finally, expectation values of the form $\corr(a_j, a_k^\dagger)_\beta$, which are just complex linear combinations of the $\corr(b_j, b_k^\dagger)_\beta$.

This allows one to calculate density-density correlations $\corr(n_j,n_k)$ via Wick's theorem.
It allows to express higher moments in terms of the second moments of the thermal states of quadratic Hamiltonians, which are Gaussian states.
Concretely, for fermionic systems we have (Lemma 6 in \cite{Gluza2016})
\begin{equation}\label{Wicksth}
  \expect{\prod _{k=1}^{m} c_{i_k}}_\beta =\text{Pf}(\Gamma[i_1,\dots, i_m]) ,
\end{equation}
where $\text{Pf}$ is the Pfaffian and $\Gamma$ has matrix elements
\begin{equation}
  \big( \Gamma[i_1,\dots,i_m] \big)_{a,b} \coloneqq \begin{cases} 
    \phantom{-}\expect{c_{i_a}\,c_{i_b}}_\beta & \text{if}\ a < b, \\
    -\expect{c_{i_b}\,c_{i_a}}_\beta & \text{if}\ a > b, \\
    \phantom{-}0 & \text{otherwise.}
  \end{cases}
\end{equation}
In particular, for the density-density correlations we find
\begin{align}\label{fourmoment}
  \hspace{-6pt}\corr_\beta(n_i,n_j) &= \expect{a_i^\dagger\,a_i\,a_j^\dagger\,a_j}_\beta - \expect{a_i^\dagger\,a_i}_\beta\,\expect{a_j^\dagger\,a_j}_\beta \\
                       &= \expect{a_i^\dagger\,a_j}_\beta\,\expect{ a_i\,a_j^\dagger}_\beta - \expect{a_i^\dagger\,a_j^\dagger}_\beta\,\expect{a_i\,a_j}_\beta . \label{fourmoment2}
\end{align}
This allows to bound density-density correlations, as well as higher order correlation functions between even and odd operators in quadratic models by means of Theorem~\ref{theorem:power-law-deca-of-correlations}.

\paragraph{Numerical analysis.}
We now present the numerical results on the decay of density-density correlations between two sites separated by a distance $l$ for different values of the chemical potentials $\mu$, inverse temperatures $\beta$ and interaction decay exponents $\alpha$.
We consider different chain lengths ($L\in\{500,1000,2000\}$) in order to identify the influence of finite size effects.
We observe that asymptotically correlations decay power-law like for any temperature and interaction strength \footnote{See Supplemental Material \ref*{powerlawdecay} for plots of the raw data.}, that is for all $i$ and large $l$
\begin{equation}
  \corr(n_i,n_{i+l}) \propto l^{-\nu},
\end{equation}
where $\nu$ characterizes the decay of the correlations.
Away from the critical point, we observe that $\nu$ depends on $\alpha$.
At the quantum critical point ($T=0$, $\mu=1$) we observe universal behavior with $\nu$ being independent of $\alpha$, namely $\nu \approx 2$. Everywhere else we find $\nu \approx 2$ when $\alpha \leq 1$ and $\nu \approx 2\,\alpha$ when $\alpha > 1$ (see Figure \ref{fig02}).
These results are in agreement with the results for the ground state in \cite{Vodola2014c}.

\begin{figure}
  \centering
  \includegraphics[width=0.9\linewidth]{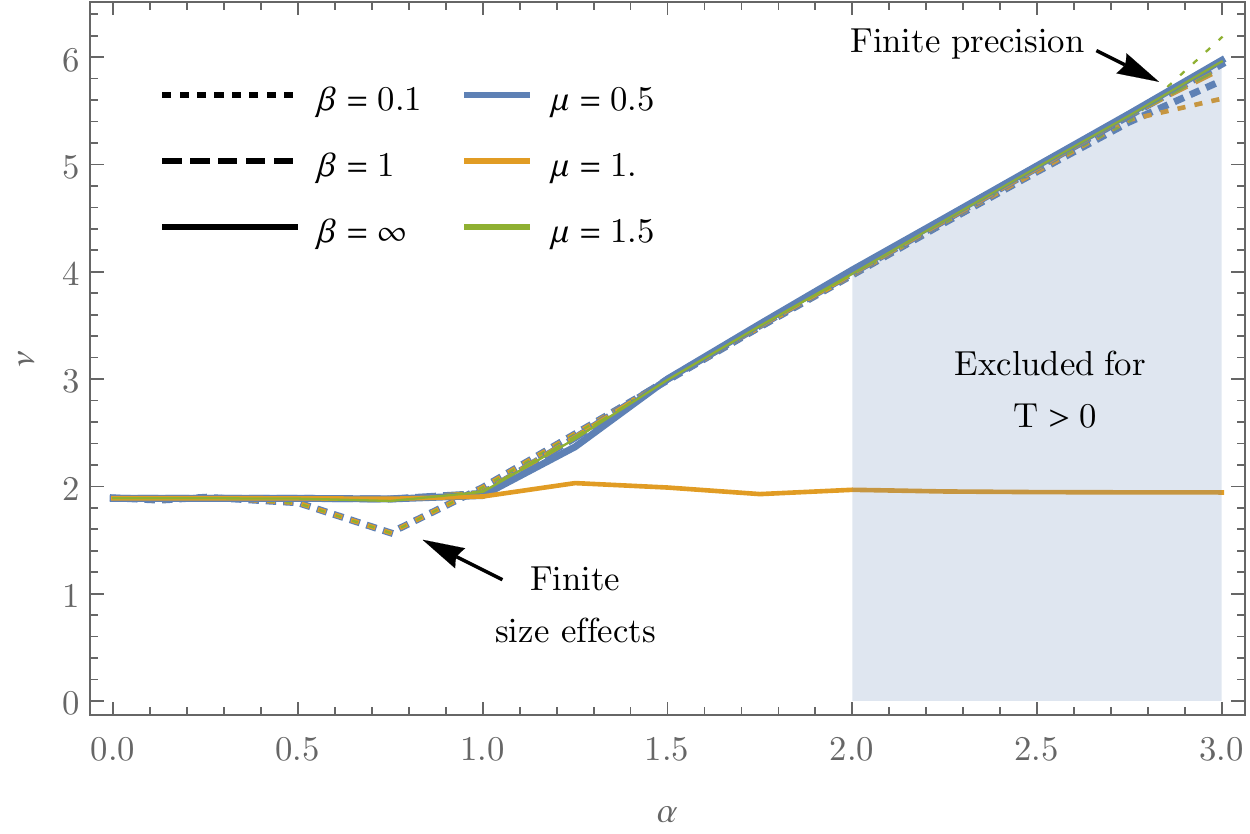}
  \caption{\label{fig02} Exponent $\nu$ as a function of the exponent of the interactions decay $\alpha$ extracted from the data for $L=2000$.
    The blue, orange, and green lines correspond to $\mu=0.5,1.0,1.5$.
    The line styles correspond to the inverse temperatures $\beta=0.1,1.0,\infty$.
    Exponents inside the shaded region are excluded by Theorem~\ref{theorem:power-law-deca-of-correlations} whenever $T>0$.
    For high temperatures and $\alpha < 1$ finite size effects slightly distort the results, for large $\alpha$ the finite precision is the limiting factor.
  }
\end{figure}

\paragraph{Application and discussion of the analytical bound.}
Let us now apply Theorem~\ref{theorem:power-law-deca-of-correlations} to the Kitaev chain.
As the model is quadratic, we can use Eq.~\eqref{fourmoment} to express the density-density correlations in terms of expectation values of odd operators and apply Theorem~\ref{theorem:power-law-deca-of-correlations}.
This yields for any $0<\epsilon<1$
\begin{align}
  \corr_\beta(n_i,n_j) \in \mathcal{O}(l^{-2(1-\epsilon)\,\alpha})
\end{align}
for any finite temperature $T>0$ and for any $\alpha > 2\,D$.

A comparison with the numerics shows that Theorem~\ref{theorem:power-law-deca-of-correlations} is asymptotically tight.
The shaded region in Fig.~\ref{fig02} is the range of decay exponents excluded by Theorem~\ref{theorem:power-law-deca-of-correlations}.
Despite the simplicity of the Kitaev chain, it shows correlations that are asymptotically as strong as possible for any fermionic system with power-law decaying two site interactions.
Further, the restriction to $T>0$ of Theorem~\ref{theorem:power-law-deca-of-correlations} is not an artifact of our proof strategy but correlations actually do decay slower at the quantum critical point at $T=0$ and $\mu=1$.

By performing a first-order high temperature expansion one can see that this model can be expected to essentially asymptotically saturate the bound from Theorem~\ref{theorem:power-law-deca-of-correlations} for $T \to \infty$.
For simplicity, consider only the long-range part of the Hamiltonian ($t=0$), then whenever correlations are analytic around $\beta=0$ (not the case for $\alpha<1$) one has in the limit $\beta\to0$
\begin{align}
  |\corr(a_1,a_j)_\beta| &= |\tr(a_1\,a_j\,\e^{-\beta\,H_{\mathrm{LR}}})/\tr(\e^{-\beta\,H_{\mathrm{LR}}})|\\
                        &\geq |\tr(a_1\,a_j\,\beta\,H_{\mathrm{LR}})/2^L - \mathcal{O}(\beta^2)| \\
                        &= |\beta\,\Delta\,d_{j-1}^{-\alpha}/4 - \mathcal{O}(\beta^2)| .
\end{align}
More generally, for an arbitrary system with local dimension $D$ and two-site interacting Hamiltonian $H \coloneqq \sum_{i,j} H_{i,j}$ and any two traceless on-site operators $A_i,B_j$ one finds that if there is an interval $[0,\beta_0]$ in which $\corr_\beta(A_i,B_j)$ is analytic, then for all $\beta \in [0,\beta_0]$
\begin{equation}
  |\corr(A_i,B_j)_\beta| \geq | \beta\,D^{-L} \tr(A_i\, B_j\, H_{i,j}) - \mathcal{O}(\beta^2) | .
\end{equation}
One expects such systems to have the strongest decay of correlations at high temperatures.
As they can essentially saturate our bound already for $\beta \approx 0$, our results indicate the absence of phase-transitions that are reflected in the asymptotic decay behavior of correlations at non-zero temperature in systems to which Theorem~\ref{theorem:power-law-deca-of-correlations} applies.

One might hope to prove a theorem similar to Theorem~\ref{theorem:power-law-deca-of-correlations} with tools from Fourier analysis \cite{Katznelson2004}.
As we discuss in \footnote{See Supplemental Material \ref*{fourieranalysis} for details of the Fourier analysis argument.}, in this way one can at most show that $\corr(a_i,a_{i+l}) \in \mathcal{O}(|l|^{-\alpha+1})$ for quadratic translation invariant Hamiltonians, which is a subclass of the systems to which Theorem~\ref{theorem:power-law-deca-of-correlations} applies.
Such a result would be weaker than Theorem~\ref{theorem:power-law-deca-of-correlations} and in particular overestimate the decay exponent, underlining the fact that our result is non-trivial.

\paragraph{Conclusions.}
We have investigated the correlation decay in systems of fermions with long-range interactions both analytically and numerically.
We have derived a general bound for the correlation decay between anti-commuting operators in systems with two site interacting Hamiltonians with power-law decaying interactions in thermal equilibrium at any $T>0$.
Our bound predicts that the correlations decay at least power-law like with essentially the same exponent as the decay of the interactions.
We have verified that our bound is asymptotically tight by a high temperature expansion and by comparing with numerical simulations of the Kitaev chain with long-range interactions and found that this model asymptotically exhibits the slowest possible decay of correlation of any fermionic model with two-site interacting and power-law decaying interactions.

\paragraph{Acknowledgements.}
We are grateful for insightful discussions with Zoltán Zimborás and Jens Eisert.
We acknowledge financial support from the European Research Council
(CoG QITBOX and AdG OSYRIS), the Axa Chair in Quantum Information Science,
Spanish MINECO (FOQUS FIS2013-46768, QIBEQI FIS2016-80773-P and Severo Ochoa Grant No.~SEV-2015-0522), Fundaci\'{o} Privada Cellex, and Generalitat de Catalunya (Grant No.~SGR 874 and 875, and CERCA Programme).
C.\ G.~acknowledges support by MPQ-ICFO, ICFOnest+ (FP7-PEOPLE-2013-COFUND), and co-funding by the European Union's Marie Skłodowska-Curie Individual Fellowships (IF-EF) programme under GA: 700140.
S.\ H.-S.~acknowledges funding from the ``laCaixa''-Severo Ochoa program and the Max Planck Prince of Asturias Award Mobility Programme.

%

\makeatletter
\newcommand{\manuallabel}[2]{\def\@currentlabel{#2}\label{#1}}
\makeatother
\manuallabel{proof_of_lemma_integralrepresentation}{Section~A}
\manuallabel{ms:proof_of_RL_bound}{Section~B}
\manuallabel{appSRPBC}{Section~C}
\manuallabel{powerlawdecay}{Section~D}
\manuallabel{fourieranalysis}{Section~E}



\cleardoublepage
\appendix
\makeatletter
\onecolumngrid
\begingroup
\frontmatter@title@above
\frontmatter@title@format
\@title{} Supplemental Material

\endgroup
\vspace{1cm}
\twocolumngrid
\makeatother
\setcounter{page}{1}

\section{\ref*{proof_of_lemma_integralrepresentation}: Proof of Lemma~\ref*{lemma:integralrepresentation}}
For the readers convenience we include a proof of our Lemma~\ref*{lemma:integralrepresentation}, which is a result from \cite{Hastings2004a} of the main text.
\begin{proof}[Proof of Lemma~\ref*{lemma:integralrepresentation}]
  Let us consider the operator $A$ with matrix elements $A_{ij}$ in some basis of eigenvectors of the Hamiltonian $H$.
  Let $E_i$, be the energy of the $i$-th eigenvector.
  Define $A_\omega$ element wise via $(A_\omega)_{ij} \coloneqq A_{ij}\,\delta(E_i-E_j-\omega)$, then $A=\int A_\omega \diff \omega$ and we can write $\expect{A_\omega\,B}_\beta = Z^{-1}\sum_{i,j}\delta(E_i-E_j-\omega)\,A_{ij}\,B_{ji}\,\e^{-\beta\,E_i}$ and similarly $\expect{B\,A_\omega}_\beta = Z^{-1}\sum_{i,j} \delta(E_i-E_j-\omega)\,B_{ji}\,A_{ij}\,\e^{-\beta\,E_j}$.
  However, $\delta(E_i-E_j-\omega)\,\e^{-\beta\,E_j} = \delta(E_i-E_j-\omega)\,\e^{-\beta\,E_i}\,\e^{\beta\,\omega}$, and thus,  $\expect{B\,A_\omega}_\beta = \expect{A_\omega\,B}_\beta\,\e^{\beta\,\omega}$.
  Hence,
  \begin{equation}\label{analyticseq03}
    \expect{ A_\omega\,B }_\beta = \frac{1}{1+\e^{\beta\,\omega}} \expect{ \left\{ A_\omega,B\right\} }_{\beta}.
  \end{equation}
  Next, we use that $(1+\e^{\beta\,\omega})^{-1} = 1/2-\beta^{-1}\sum_{n\text{   odd}}(\omega-\ii\,n\,\pi/\beta)^{-1}$, where the sum ranges over all positive and negative odd $n$.
  For $n>0$, we have $(\omega -\ii\,n\,\pi/\beta)^{-1}=\ii\int_{0}^{\infty} \e^{-(\ii\,\omega+n\,\pi/\beta)\,t}\diff t$.
  Similarly, for $n<0$, we have $(\omega -\ii\,n\,\pi/\beta)^{-1}=-\ii\int_0^{\infty} \e^{(\ii\,\omega+n\,\pi/\beta)\,t}\diff t$.
  Thus,
  \begin{equation}\label{analyticseq04}
    \frac{1}{1+\e^{\beta\,\omega}} = \frac{1}{2}+\frac{\ii}{\beta} \int_0^{\infty} \frac{\e^{\ii\,\omega\,t}- \e^{-\ii\,\omega\,t}}{\e^{\pi\,t/\beta}-\e^{-\pi\,t/\beta}}\diff t .
  \end{equation}
  Due to the linearity of time evolution $A(t) \coloneqq \e^{\ii\,H\,t}\,A\,\e^{-\ii\,H\,t}$ we have $A_\omega(t)=\e^{\ii\,\omega\,t}\,A_\omega$.
  Therefore, substituting Eq.~\eqref{analyticseq04} into Eq.~\eqref{analyticseq03}, we get
  \begin{equation}\label{analyticseq05}
    \begin{split}
      \expect{A_\omega\,B}_\beta &= \frac{1}{2} \expect{\{A_\omega,B\}}_\beta \\
      &+ \frac{\ii}{\beta} \int_0^\infty \frac{\expect{\{A_\omega(t) - A_\omega(-t), B\}}_\beta}{\e^{\pi\,t/\beta}-\e^{-\pi\,t/\beta}} \diff t .
    \end{split}
  \end{equation}
  Finally, by integrating Eq.~\eqref{analyticseq05} over $\omega$ we get Eq.~\eqref{eq:integralrepresentation}.
\end{proof}

\section{\ref*{ms:proof_of_RL_bound}: Proof of Lemma~\ref{lemma:long-range-lieb-robinson}}
Here we discuss how to prove Lemma~\ref{lemma:long-range-lieb-robinson} following the strategy outlined in \cite{Foss-Feig2014} of the main text.
\begin{proof}[Proof of Lemma~\ref{lemma:long-range-lieb-robinson}]
  As in \cite{Foss-Feig2014} of the main text the Hamiltonian is separated into a finite-range and a long-range part.
  All interactions over distances up to some length $\chi$ go into the finite-range part of the Hamiltonian
  \begin{equation}
    H_{\mathrm{FR}} \coloneqq \sum_{\kappa,i,j\colon d_{i,j}\leq\chi} J_{i,j}^{(\kappa)}\,V_i^{(\kappa)}\,V_j^{(\kappa)}
  \end{equation}
  and all others into the long-range part.
  As $\chi$ is later chosen to grow with time, one should think of both parts of the Hamiltonian as piece wise constant in time.
  For any operator $A$ let $\mathcal{A}(t)$ be the time evolution of $A$ under $H_{\mathrm{FR}}$ only.
  Due to standard Lieb-Robinson bounds the time evolution under the finite-range part is quasi-local, i.e., $\mathcal{A}(t)$ can be decomposed into a sum of operators $\sum_l^\infty \mathcal{A}^l(t)$, each supported only on the support of $A$ and a border of width $\chi\,l$ around it.
  The norm of these operators can be bounded proportional to $\exp(v(\chi)\,t-l)$ with the speed
  \begin{align}
    v (\chi) &\coloneqq 4\,\exp(1)\,\sup_i \sum_{\kappa,j\colon d_{i,j} \leq \chi} J_{i,j}^{(\kappa)}\\
             &\leq 4\,\exp(1)\,J\,2^D \sum_{d=1}^\chi d^{D-1-\alpha} .
  \end{align} 
  It is crucial that the speed $v(\chi)$ of the finite-range Hamiltonian $H_{\mathrm{FR}}$ can be bounded independently of $\chi$ by $v(\chi) \leq v \coloneqq 4\,J\,\exp(1)\,2^D\,\zeta(1+\alpha-D) \leq 8\,J\,\exp(1)\,2^D$, where $\zeta$ is the Riemann zeta function and we have used that $\alpha > 2\,D$.
  In particular Eqs.~(S3) and (S8) from the Supplemental Information of \cite{Foss-Feig2014} from the main text also hold in our setting with an anti-commutator instead of a commutator.
  One then makes use of the interaction picture to bound the additional growth of the support due to the long-range part.
  We define $C_r(t) \coloneqq \|\{A(t),B\}\|$ in analogy to the quantity introduced in Eq.~(S9) of the Supplementary Information of \cite{Foss-Feig2014} from the main text and proceed as in Section~S2.
  Let $\mathcal{U}(t)$ be the interaction picture unitary, i.e., the unitary for which for any operator $A$ it holds that $A(t) = \mathcal{U}^\dagger(t)\,\mathcal{A}(t)\,\mathcal{U}(t)$.
  The idea is now to introduce the generalized (two-time) anti-commutator
  \begin{equation}
    C_r^l(t,\tau) \coloneqq \{ \mathcal{A}^l(t), \mathcal{U}(t)\,B\,\mathcal{U}^\dagger(t)\} 
  \end{equation}
  (instead of the commutator) with the property that $\|\sum_l C_r^l(t,t)\| = C_r(t)$.
  By using the von Neumann equation and the equality
  \begin{equation}
    \{ A, [B,C] \} = \{ C, [A,B] \} + [B, \{ C,A \}] ,
  \end{equation}
  (instead of the Jacobi identity) one obtains a differential equation for $C_r^l(t,\tau)$ equivalent to Eqs.~(S11) and (S16) from \cite{Foss-Feig2014} of the main text with the outer commutator in the second term replaced by an anti-commutator.
  After employing the bound (S17), also in the fermionic case, a part of the right hand side can be identified to be $C_r^l(t)$ allowing for the same type of recursive bound on $\|C_r^l(t,t)\|$.
  As everything is now reduced to scalars, one can proceed completely analogous to the proof in \cite{Foss-Feig2014} of the main text to obtain, with $c_0$, $c_1$, and $\vartheta$ constants,
  \begin{equation} 
    C_r(t) \leq c_0 \e^{v\,t-r/\chi} + c_1 \e^{v_\chi\,t} (\chi\,v\,t/r)^\alpha
  \end{equation}
  where
  \begin{equation}
    v_\chi \leq \vartheta\,t^D\,\chi^{2\,D-\alpha} ,
  \end{equation}
  which is the analogue to Eq.~(18) in \cite{Foss-Feig2014} of the main text.

  That Lemma~\ref{lemma:long-range-lieb-robinson} is restricted to $\alpha > 2\,D$ is a consequence of the above bound on $v_\chi$, which becomes small for large $\chi$ only if $\alpha > 2\,D$.
  It can be shown to hold as follows.

  The quantity $v_\chi$ is defined as $v_\chi \coloneqq \vartheta'\,(\chi\,v\,t)^D\,\lambda_\chi$ with
  \begin{align}
    \lambda_\chi &\leq \sum_{d=\chi+1}^\infty J\,d^{-\alpha}\,2\,(2\,d)^{D-1} \\
                 &= 2^D\,J\,\sum_{d=\chi+1}^\infty d^{D-\alpha-1} \\
                 &= 2^D\,J\,\zeta(\alpha-D+1,\chi+1) 
  \end{align}
  where $\zeta$ is the Hurwitz zeta function (a generalization of the Riemann zeta function).
  In total this gives
  \begin{equation}
    v_\chi \leq \vartheta'\,(2\,\chi\,v\,t)^D\,J\,\zeta(\alpha-D+1,\chi+1) ,
  \end{equation}
  and it remains to show a bound on $\zeta$ for large $\chi$.
  We make use of the following integral representation of $\zeta$, valid for all $\alpha-D+1>0$ and $\chi>0$:
  \begin{equation}
    \zeta(\alpha-D+1,\chi) = \Gamma(\alpha-D+1)^{-1}\,\int_0^\infty \frac{x^{\alpha-D}\,\e^{-\chi\,x}}{1-\e^{-x}} \diff x
  \end{equation}
  The integrand can be bounded using 
  \begin{equation}
    \frac{1}{e^{x/2}-e^{-x/2}} \leq \frac{1}{x} \Longrightarrow \frac{1}{1-e^{-x}} \leq \frac{e^{x/2}}{x} ,
  \end{equation}
  which, as long as $\alpha>D$, allows to compute the resulting integral explicitly
  \begin{align}
    &\int_0^\infty x^{\alpha-D-1}\,\e^{-(\chi-1/2)\,x} \diff x\\
    ={} & \Gamma(\alpha-D)\,(\chi-1/2)^{D-\alpha} ,
  \end{align}
  which yields the following bound on $v_\chi$
  \begin{equation}
    v_\chi \leq \vartheta'\,(2\,v)^D\,J\,\frac{\Gamma(\alpha-D)}{\Gamma(\alpha-D+1)}\,t^D\,\chi^D\,(\chi-1/2)^{D-\alpha} .
  \end{equation}

  Proceeding as in \cite{Foss-Feig2014} of the main text one obtains Lemma~\ref{lemma:long-range-lieb-robinson} as stated in the main text.
\end{proof}

\section{\ref*{appSRPBC}: PBC implies short-range interactions}
Here we show that the long-range contribution of the Hamiltonian \eqref{eq:hamiltonianLR},
\begin{equation}
  H_{\mathrm{LR}} \coloneqq \frac{\Delta}{2} \sum_{i=1}^L \sum_{j=1}^{L-1}d_j^{-\alpha}\left( a_i\, a_{i+j} + a_{i+j}^\dagger\,a_i^\dagger \right) ,
\end{equation}
is only non-negligible when antiperiodic boundary conditions are considered, that is, for $i>L$ we set $a_{i} \coloneqq -a_{i \mod L}$.
For simplicity, we study the problem for both periodic and antiperiodic boundary conditions and make use of a parameter $p$ which characterizes the boundary conditions, such that $p=+1$ corresponds to periodic boundary conditions (PBC) and $p=-1$ corresponds to antiperiodic (ABC).

First, we analyze the first term of the long-range term \eqref{eq:hamiltonianLR} (the annihilation-annihilation term) and divide the sum in $j$ into two contributions, such as
\begin{align}
  H_{\mathrm{LR}}^{\mathrm{a-a}}&=\sum_{i=1}^L \sum_{j=1}^{L-1}d_j^{-\alpha}\,a_i\,a_{i+j} \\
  \label{eqLRaa}
                                &= \sum_{i=1}^L \left( \sum_{j=1}^{L-i}d_j^{-\alpha}\,a_i\,a_{i+j}+\sum_{j=L-i+1}^{L-1}d_j^{-\alpha} a_i\,a_{i+j} \right).
\end{align}

Given this, we apply the boundary conditions and introduce a change of indexes $j' = i+j-L$ for the second term, such that
\begin{equation}
  \sum_{i=1}^{L} \sum_{j'=1}^{i-1} d_{j'+L-i}^{-\alpha}\,a_i\,a_{j'+L} = p \sum_{i=1}^{L} \sum_{j'=1}^{i-1} d_{j'+L-i}^{-\alpha}\,a_i\,a_{j'}.
\end{equation}

Then we reorder the sums as follows,
\begin{equation}
  p \sum_{i=1}^{L} \sum_{j'=1}^{i-1} d_{j'+L-i}^{-\alpha}\,a_i\,a_{j'}=p \sum_{j'=1}^{L} \sum_{i=j'+1}^{L} d_{j'+L-i}^{-\alpha}\,a_i\,a_{j'}.
\end{equation}
We apply the canonical commutation relation $\lbrace a_i,a_k \rbrace =0$, make two changes of indexes: first, $j'' \rightarrow i-j'$ and, second, $i\rightarrow j'$ and $j\rightarrow j''$; and apply $d_{L-j}=d_j$. Finally, we get
\begin{align}
  p \sum_{j'=1}^{L} \sum_{i=j'+1}^{L} d_{j'+L-i}^{-\alpha}\,a_i\,a_{j'}&=-p \sum_{j'=1}^{L} \sum_{j''=1}^{L-j'} d_{L-j''}^{-\alpha}\,a_{j'}\,a_{j'+j''}\\
  \label{eqLRaa2}
                                                                       &=-p \sum_{i=1}^{L} \sum_{j=1}^{L-i} d_{j}^{-\alpha}\,a_{i}\,a_{i+j}.
\end{align}
We substitute the equation \eqref{eqLRaa2} into the term \eqref{eqLRaa}, such that
\begin{equation}\label{eqLRaasolved}
  H_{\mathrm{LR}}^{\mathrm{a-a}} = \sum_{i=1}^L \left( \sum_{j=1}^{L-i}d_j^{-\alpha}\,a_i\,a_{i+j}-p \sum_{j=1}^{L-i} d_{j}^{-\alpha}\,a_{i}\,a_{i+j} \right)\text{.}
\end{equation}

Given this expression, it is clear that this term and its conjugate cancel for periodic boundary conditions.
We can conclude then that the long-range term does not contribute for PBC and for any interaction exponent $\alpha$.
On the other hand, the long-range term \eqref{eq:hamiltonianLR} for antiperiodic boundary conditions is not null and can be reexpressed as
\begin{equation}
  H_{\mathrm{LR}}=\Delta\,\sum_{i=1}^{L} \sum_{j=1}^{L-i} d_j^{-\alpha}\,\left( a_i\,a_{i+j}+a_{i+j}^{\dagger}\,a_i^{\dagger} \right).
\end{equation}

\onecolumngrid
\newpage
\section{\ref*{powerlawdecay}: Power-law decay of correlations in the Kitaev chain}
\begin{figure*}[h]
  \includegraphics[height=0.23\textheight]{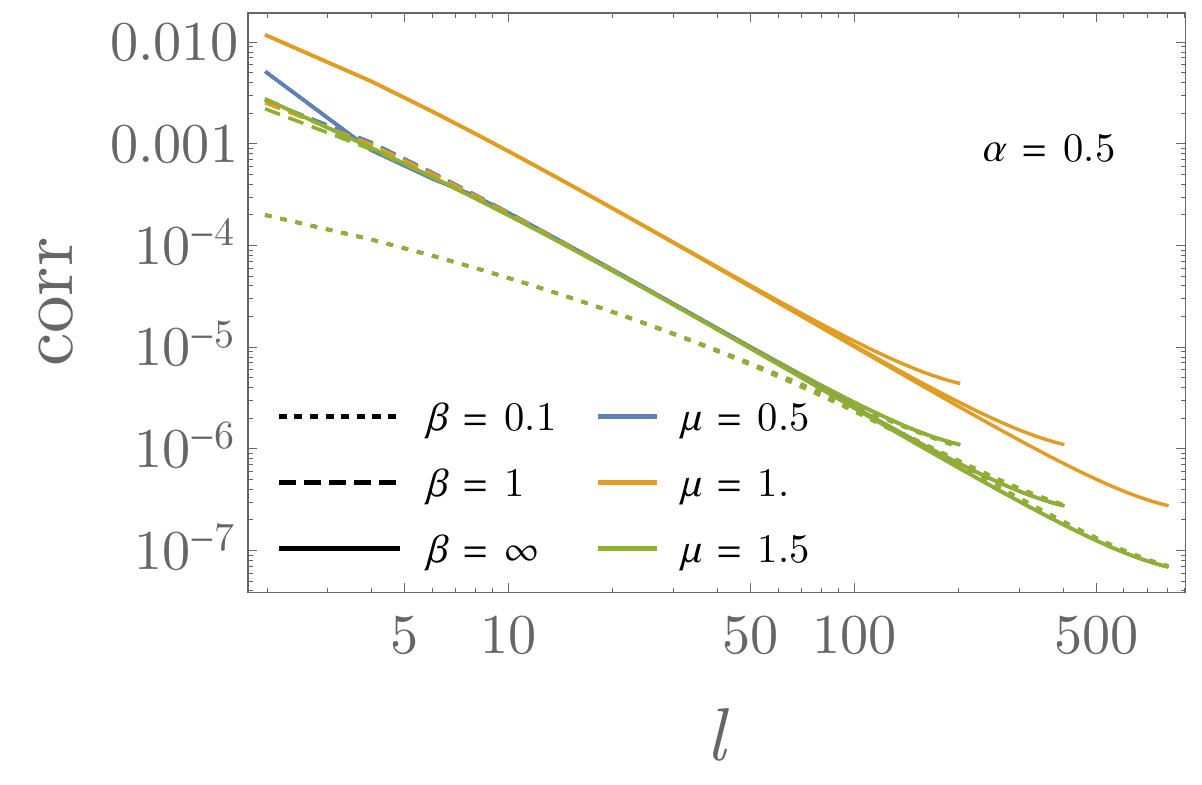}\hfill\includegraphics[height=0.23\textheight]{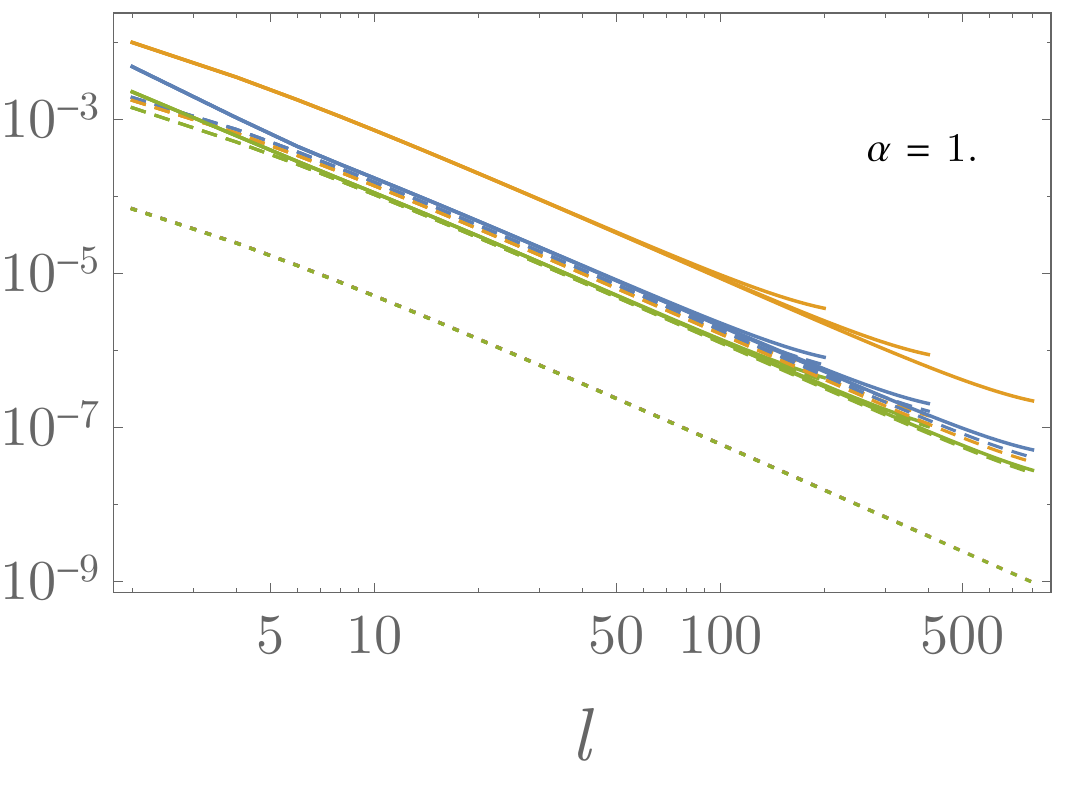}\\\hspace{8mm}\includegraphics[height=0.23\textheight]{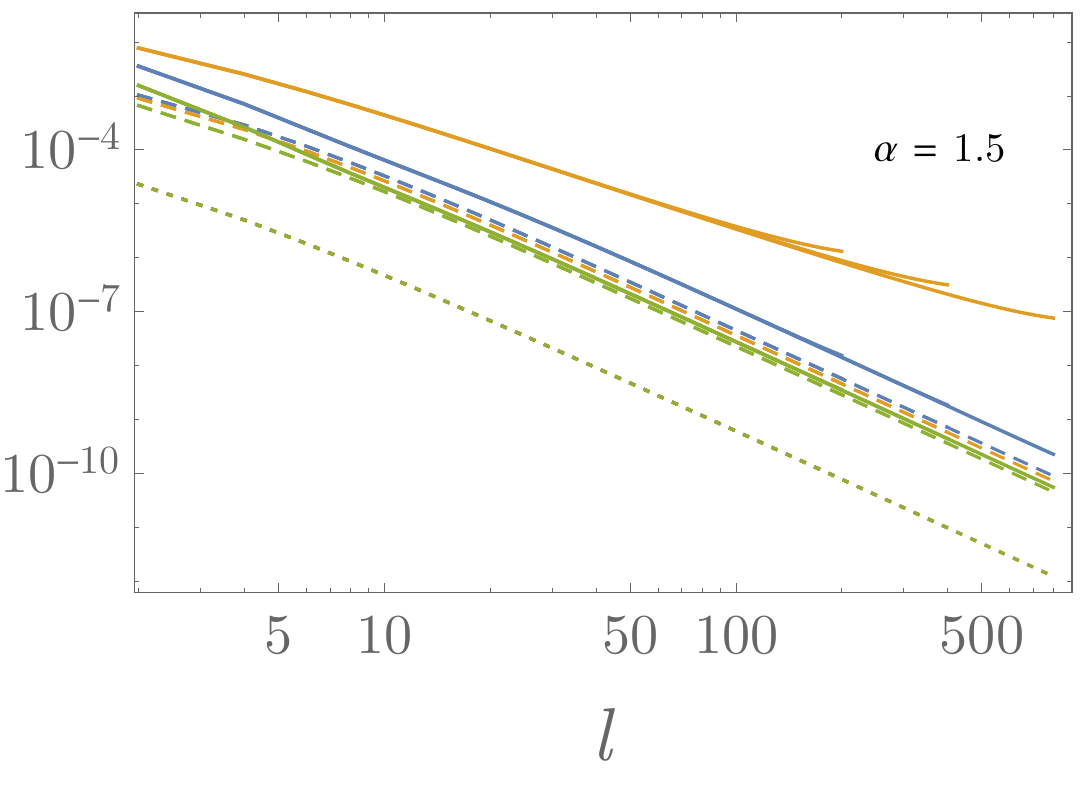}\hfill\includegraphics[height=0.23\textheight]{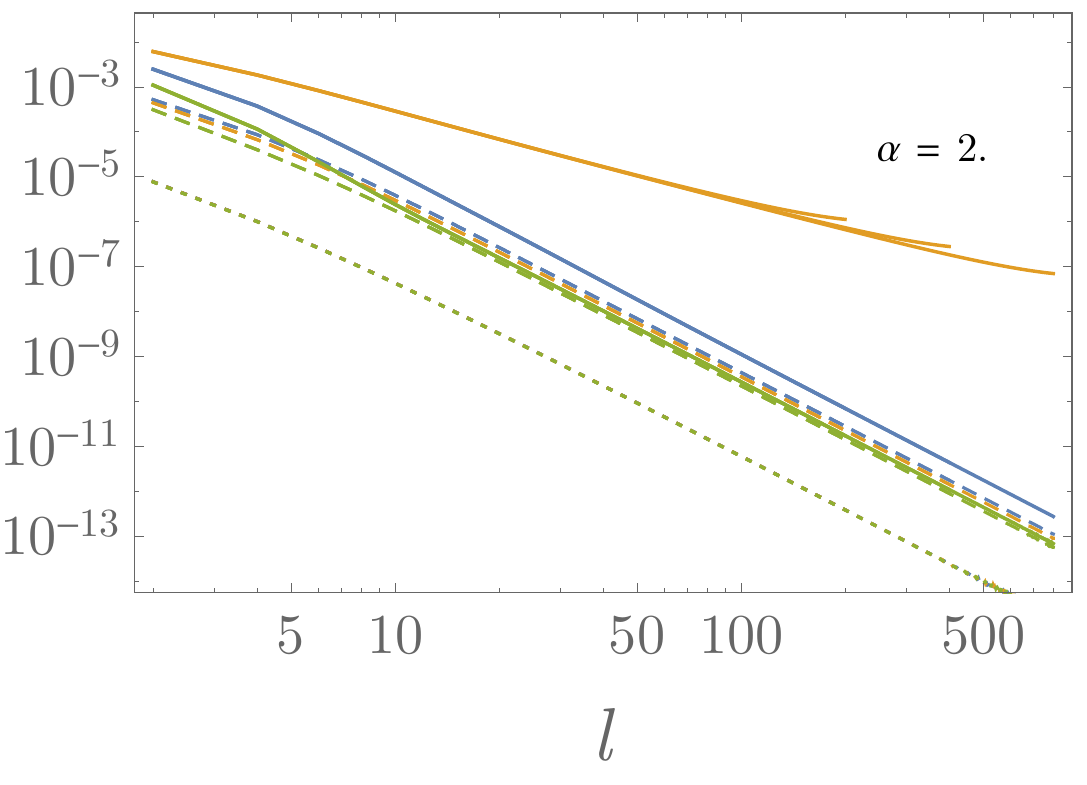}
  \caption{\label{fig01} Double logarithmic plots of the correlations $\corr$ as a function of the distance $l$ for $\alpha=0.5,1,1.5,2$ from left to right.
    The blue, orange and green lines correspond to $\mu=0.5,1.0,1.5$.
    The different line styles indicate different inverse temperatures, namely $\beta=0.1,1.0,\infty$ respectively.
    For each combination we overlay curves for chain lengths $L=500,1000,2000$ to visualize the influence of finite size effects.
    For large $\alpha$ and high temperatures a bending of the curves at short distances is visible, reminiscent of the transient behavior observed for $\alpha>1$ at $T=0$ in \cite{Vodola2014c} of the main text.
    The exponents shown in Figure~\ref{fig02} in the main text were determined by linear fits to the logarithmized data in the range $l\in[l_{\mathrm{min}},300]$ with $l_{\mathrm{min}} = 200$ except for $\alpha\geq2$, where $l_{\mathrm{min}}=50$ for $\beta=1,\infty$ and $l_{\mathrm{min}}=20$ for $\beta=0.1$.
    Data with $\corr < \e^{-32}$ were discarded.
    The remaining data is almost perfectly linear in the double logarithmic plot.
  }
\end{figure*}

\twocolumngrid
\section{\ref*{fourieranalysis}: Fourier analysis}
Here we compare our result with what can be obtained using tools from Fourier analysis.
It is known that one can essentially show the following (some additional conditions omitted for the sake of brevity, see \cite[Section I.4]{Katznelson2004} of the main text for more details):
\renewcommand{\theenumi}{\roman{enumi}}%
\begin{enumerate}
\item If the absolute values $|f_k|$ of the Fourier coefficients of a function $f$ decay slightly faster than $|k|^{-\alpha}$, then $f$ is almost $(\alpha-1)$-times continuously differentiable.
\item If a function $f$ is $\alpha'$-times continuously differentiable, then the absolute values $|f_k|$ of its Fourier coefficients decay like $|k|^{-\alpha'}$.
\end{enumerate}
\renewcommand{\theenumi}{enumi}%
If the Hamiltonian $H$ of a 1D long range system is quadratic and translation invariant, then the Hamiltonian matrix $h_{ij} \in \mathcal{O}(|i-j|^{-\alpha})$ is circulant and its first row can be thought of as the Fourier coefficients of a function $f(x)$ that is almost $(\alpha-1)$-times continuously differentiable. In turn, $\corr(a_i,a_{i+l})_\beta$ can be thought of as the Fourier coefficients of the function $g(x) \coloneqq 1/(1+\e^{\beta\,f(x)})$, which is also almost $(\alpha-1)$-times continuously differentiable, and thus $\corr(a_i,a_{i+l}) \in \mathcal{O}(|l|^{-\alpha+1})$.

\end{document}